\documentclass[12pt,leqno,letterpaper]{article}

\usepackage{amsmath,amsthm,enumerate,amssymb}
\usepackage[utf8]{inputenc}

\usepackage[english]{babel}
\usepackage{graphicx}
\usepackage{color}

\usepackage{graphicx}

\newtheorem{theorem}{Theorem}[section]

\newtheorem{corollary}[theorem]{Corollary}

\newcommand{\tr}{{\rm Tr\hskip -0.2em}~}

          \newcommand{\df}[2]{\frac{d#1}{d#2}}
          
\begin{document}

\title{Reduced relative quantum entropy}
\author{Frank Hansen}
\date{September 11, 2022}

\maketitle

\begin{abstract} We introduce the notion of reduced relative quantum entropy and prove that it is convex.
This result is then used to give a simplified proof of a theorem of Lieb and Seiringer\\[1ex]
\end{abstract}

\section{Reduced relative entropy}

Let $ H $ be a contraction. We introduce the reduced relative quantum entropy $ S_H(A,B) $ for positive definite matrices $ A $ and $ B $ by setting
\[
S_H(A\mid B)=\tr\bigl(A\log A-H^*AH\log B-A+B\bigr).
\]
Note that for $ H=I $ (the identity matrix) we recover the relative quantum entropy
\[
S(A\mid B)=S_I(A\mid B)=\tr\bigl(A\log A-A\log B-A+B\bigr).
\]

\begin{theorem}
The reduced relative quantum entropy is convex.
\end{theorem}

\begin{proof}

Let $ H $ be a contraction, and let $ A $ and $ B $ be positive definite matrices. By Lieb's concavity theorem  \cite[Corollary 1.1]{lieb:1973:1}, cf. also  \cite[Corollary 2.2]{hansen:2006:3},  the function
\[
(A,B)\to \tr H B^p H^* A^{1-p}\qquad 0\le p\le 1
\]
is concave. The function
\[
(A,B)\to \frac{\tr HB^pH^* A^{1-p}-\tr HH^*A}{p}
\]
is therefore also concave for $ 0<p\le 1. $ By letting $ p $ tend to zero we obtain that the limit function
\[
(A,B)\to \df{}{p} \tr HB^pH^* A^{1-p}\Bigr|_{p=0}
\]
is concave. Since by calculation
\[
 \df{}{p} \tr HB^pH^* A^{1-p}=\tr\bigl(HB^p \log (B)H^* A^{1-p}-HB^pH^*A^{1-p}\log A\big)
 \]
we obtain
\[
 \df{}{p} \tr HB^pH^* A^{1-p}\Bigr|_{p=0}=\tr\bigl(H\log (B) H^*A -HH^*A\log A\bigr).
\]
The map
\begin{equation}\label{first concave map}
(A,B)\to \tr\bigl(H^*AH\log B -HH^*A\log A\bigr)
\end{equation}
is thus concave, where we used the cyclicity of the trace. We then write the reduced quantum relative entropy on the form
\[
\begin{array}{l}
S_H(A\mid B)=\tr\bigl(A\log A-H^*AH\log B-A+B\bigr)\\[2ex]
=-\tr\bigl(H^*AH\log B-HH^*A\log A\bigr)+\tr\bigl((I-HH^*)A\log A-A+B\bigr).
\end{array}
\]
The first term is convex by concavity of the map in (\ref{first concave map}). The second term is 
convex since $ HH^*\le I $ and $ A\to A\log A $ is convex. 
\end{proof}

The relative quantum entropy $ S(A\mid B) $ is non-negative with equality for $ A=B. $ Therefore,
\[
\tr B=\max_{X>0}\tr(X\log B-X\log X+X)
\]
for any positive operator $ B, $ cf. \cite[Lemma 6]{tropp:2012:2}. If we to a contraction $ H $ put 
\[
B=\exp\bigl(L+H^*\log(A)H\bigr),
\]
where $ L $ is self-adjoint, we obtain
\[
\begin{array}{l}
\tr \exp\bigl(L+H^*\log(A)H\bigr)\\[2ex]
=\displaystyle\max_{X>0} \tr\bigl( X(L+H^*\log(A)H)-X\log X+X\bigr)\\[2ex]
=\displaystyle\max_{X>0}\bigl\{-\tr(X\log X-HXH^*\log A-X+A)+\tr (XL+A)\bigr\}\\[2ex]
=\displaystyle\max_{X>0}\bigl\{-S_{H^*}(X\mid A)+\tr(XL+A)\bigr\}.
\end{array}
\]
Since the function
\[
(X,A)\to -S_{H^*}(X\mid A)+\tr(XL+A)
\]
is jointly concave, it follows by a well-known theorem \cite[Lemma 2.3]{carlen:2008} that the partial maximisation  over $ X $ is concave. We have thus proved.

\begin{theorem}\label{second concave map}
Let $ H $ be a contraction and $ L $ self-adjoint. The function
\[
\varphi(A)=\tr \exp\bigl(L+H^*\log(A)H\bigr)
\]
is concave in positive definite matrices.
\end{theorem}

Finally, we recover a theorem of Lieb and Seiringer \cite[Theorem 3]{lieb:2005}.

\begin{corollary}\label{Concavity corollary 1}
Let $ L $ be a self-adjoint $ n\times n $ matrix and consider $ m\times n $ matrices  $ H_1,\dots,H_k $ with
\[
H_1^*H_1+\cdots+H_k^*H_k\le I_n
\]
where $ I_n $ denotes the $ n\times n $ unit matrix. The trace function
\begin{equation}\label{trace function of several variables}
\varphi(A_1,\dots,A_k)=\tr\exp\bigl(L+H^*_1 \log (A_1) H_1+\cdots+ H^*_k \log (A_k) H_k\bigr)
\end{equation}
is concave in $ k $-tuples of positive definite $ m\times m $ matrices.  
\end{corollary}

\begin{proof}
We consider the $ k\times k $ block matrices
\[
\hat A=\begin{pmatrix}
     A_1     & 0     & \cdots   & 0\\
     0         & A_2 &             & 0\\
     \vdots &         & \ddots  & \vdots\\
     0        & 0       & \dots    & A_k
     \end{pmatrix},\,
\hat L=\begin{pmatrix}
     L     & 0     & \cdots   & 0\\
     0         & 0 &             & 0\\
     \vdots &         & \ddots  & \vdots\\
     0        & 0       & \dots    & 0
     \end{pmatrix},\,
\hat H=\begin{pmatrix}
     H_1    & 0        & \cdots & 0\\
     H_2    & 0        & \cdots & 0\\
     \vdots & \vdots & \ddots & \vdots\\
     H_k    & 0         & \cdots & 0
     \end{pmatrix}
\]
with zero matrices of suitable orders inserted and note that $ H $ is a contraction. Furthermore,
\[
\hat L+\hat H^*( \log \hat A) \hat H=\begin{pmatrix}
     \displaystyle L+\sum_{i=1}^k H_i^* (\log A_i) H_i    & 0       & \cdots   & 0\\
     0                                                     & 0        &  \cdots  & 0\\
     \vdots                                             & \vdots & \ddots   & \vdots\\
     0                                                     & 0        & \dots     & 0
     \end{pmatrix}.
\]
Thus
\[
\tr\exp\bigl( \hat L+\hat H^*\log (\hat A) \hat H\bigr)=\tr\exp\Bigl(L+\sum_{i=1}^k H^*_i \log (A_i) H_i\Bigr) +(k-1)n
\]
and the statement follows from the preceding Theorem.
\end{proof}

\subsection{Interpolation between GT and Jensen}

Note that the trace function in (\ref{trace function of several variables}) becomes positively homogeneous if
\[
H_1^*H_1+\cdots+H_k^*H_k= I_n\,. 
\]
Under this extra condition we may now recover the inequality
\begin{equation}\label{GT and Jensen}
\tr\exp\Bigl(L+\sum_{i=1}^k H^*_i B_i H_i\Bigr)\le\tr\exp\bigl(L) \sum_{i=1}^k H_i^* \exp (B_i) H_i
\end{equation}
for self-adjoint $ n\times n $ matrices $ L $ and self-adjoint $ m\times m $ matrices $ B_1,\dots, B_k $ by using Corollary~\ref{Concavity corollary 1} and following the steps in \cite[Theorem 4.2]{hansen:2015:2}. This is for $ m=n $ the same bound as obtained when all the matrices commute. We are thus allowed to estimate partition functions or the Helmhotz function in quantum statistical mechanics and obtain bounds on the same form as they appear in classical physics. 

Note that (\ref{GT and Jensen}) reduces to the Golden-Thompson inequality for $ k=1, $ $ m=n, $ and $ H_1=I_n $  and to convexity under the trace of the exponential function for $ L=0. $ The inequality may thus be considered as an interpolation inequality between Golden-Thompson's inequality and Jensen's trace inequality \cite[Theorem 2.4]{hansen:2003:2}, cf. also \cite{hansen:2003:3}. However, we cannot derive  (\ref{GT and Jensen}) from these special cases. If we first apply Golden-Thompson's inequality then we obtain
\[
\tr\exp\Bigl(L+\sum_{i=1}^k H^*_i B_i H_i\Bigr)\le\tr\exp\bigl(L)\exp\sum_{i=1}^k H_i^* B_i H_i
\]
but this inequality is insufficient to obtain (\ref{GT and Jensen}), since  $ L $ is arbitrary and the exponential function is not operator convex.

{\small


\begin{thebibliography}{1}

\bibitem{carlen:2008}
E.A. Carlen and E.H. Lieb.
\newblock A \uppercase{M}inkowsky type trace inequality and strong
  subadditivity of quantum entropy \uppercase{II}: \uppercase{C}onvexity and
  concavity.
\newblock {\em Lett. Math. Phys.}, 83:107--126, 2008.

\bibitem{hansen:2006:3}
F.~Hansen.
\newblock Extensions of \uppercase{L}ieb's concavity theorem.
\newblock {\em Journal of Statistical Physics}, 124:87--101, 2006.

\bibitem{hansen:2003:2}
F.~Hansen and Pedersen G.K.
\newblock Jensen's operator inequality.
\newblock {\em Bull. London Math. Soc.}, 35:553--564, 2003.

\bibitem{hansen:2003:3}
F.~Hansen and Pedersen G.K.
\newblock Jensen's trace inequality in several variables.
\newblock {\em International Journal of Mathematics}, 14:667--681, 2003.

\bibitem{hansen:2015:2}
Frank Hansen.
\newblock Multivariate extensions of the
  \uppercase{G}olden-\uppercase{T}hompsen inequality.
\newblock {\em Annals of Functional Analysis}, 6(4):301--310, 2015.

\bibitem{lieb:1973:1}
E.~Lieb.
\newblock Convex trace functions and the
  \uppercase{W}igner-\uppercase{Y}anase-\uppercase{D}yson conjecture.
\newblock {\em Advances in Math.}, 11:267--288, 1973.

\bibitem{lieb:2005}
Elliott~H. Lieb and Robert Seiringer.
\newblock Stronger subadditivity of entropy.
\newblock {\em Physical Review A}, 71:062329, 2005.

\bibitem{tropp:2012:2}
J.A. Tropp.
\newblock From joint convexity of quantum relative entropy to a concavity
  theorem of \uppercase{L}ieb.
\newblock {\em Proceedings of the American Mathematical Society},
  140(5):1757--1760, 2012.

\end{thebibliography}

\vfill

\noindent Frank Hansen: Department of Mathematical Sciences. Copenhagen University. Email: frank.hansen@math.ku.dk

}

\end{document}